\documentclass[11pt,reqno]{amsart}

\usepackage{amssymb}
\usepackage{amsthm}
\usepackage{color}
\newtheorem {theorem} {Theorem}

\newtheorem {proposition} [theorem]{Proposition}

\newtheorem {lemma}  [theorem]{Lemma}

\newcommand{\R}{\ensuremath{\mathbb{R}}}

\newcommand{\N}{\ensuremath{\mathbb{N}}}

\newcommand{\pd}[2]{\dfrac{\partial #1}{\partial #2}}

\advance\textwidth by +1.0in
\advance\textheight by +0.5in
\advance\oddsidemargin by -0.5in
\advance\evensidemargin by -0.5in
\advance\topmargin by -0.5in

\title[Analytic integrability of Bianchi models]
{Analytic integrability of the Bianchi Class {\sc A} cosmological
models with $0\leq k<1$}

\thanks{All the authors are partially supported by the MICINN/FEDER grant
MTM2008-03437. A.F. is additionally supported by grants {\it Juan de
la Cierva}, 2009SGR410 and MTM2009-14163-C02-02. J.L. is additionally
partially supported by an AGAUR grant number 2009SGR410 and by ICREA
Academia. C.P. is additionally partially supported by the MICINN/FEDER
grant number MTM2009-06973 and by the AGAUR grant number 2009SGR859.
}

\author[A. Ferragut, J. Llibre and C. Pantazi]
{Antoni Ferragut$^1$, Jaume Llibre$^2$ and Chara Pantazi$^3$}

\address{$^1$ Departament de Matem\`{a}tica Aplicada I, Universitat Polit\`ecnica de Cata\-lunya, ETSEIB, Av. Diagonal, 647, 08028, Barcelona, Catalonia, Spain} \email{Antoni.Ferragut@upc.edu}

\address{$^2$ Departament de Matem\`{a}tiques, Universitat Aut\`{o}noma de Barcelona, Edi\-fici C, 08193 Bellaterra, Barcelona, Catalonia, Spain}
\email{jllibre@mat.uab.cat}

\address{$^3$ Departament de Matem\`atica Aplicada I, Universitat Polit\`ecnica de Cata\-lunya, EPSEB, Av. Doctor Mara\~{n}\'on, 44--50, 08028 Barcelona, Catalonia, Spain}
\email{chara.pantazi@upc.edu}

\subjclass[2010]{34A05, 34A34, 34C14}

\keywords{homogeneous systems, polynomial first integral, analytic
first integral, Bianchi cosmological models}

\begin{document}

\begin{abstract}
There are many works studying the integrability of the Bianchi class
{\sc A} cosmologies with $k=1.$ Here we characterize the analytic
integrability of the Bianchi class {\sc A} cosmological models when
$0\leq k<1$.
\end{abstract}

\maketitle


\section{Introduction}\label{s1}
Bianchi models describe space-times which are foliated by
homogeneous (and so we have three dimensional Lie algebras)
hypersurfaces of constant time. Bianchi \cite{B1, B2} was the first
to classify three dimensional Lie algebras which are nonisomorphic.
There are nine types of models according to the dimension $n$ of the
derivative subalgebra:
\begin{enumerate}
\item[(a)] $n=0$: type I;
\item[(b)] $n=1$: types II, III;
\item[(c)] $n=2$: types IV, V, VI, VII;
\item[(d)] $n=3$: types VIII, IX.
\end{enumerate}

If we consider $X_1,X_2,X_3$ an appropriate  basis of the
3-dimensional Lie Algebra, then the classification depends on a
scalar $a\in\R$ and a vector $(n_1,n_2,n_3)$, with
$n_i\in\{+1,-1,0\}$ such that
\[
[X_1,X_2]=n_3X_3, \quad [X_2,X_3]=n_1X_1-aX_2, \quad
[X_3,X_1]=n_2X_2+aX_1,
\]
where $[,]$ is the Lie bracket. In particular for $a=0$ we obtain
models of class A and for $a\neq 0$  we obtain models of class B. A
good reference for the Bianchi models is Bogoyavlensky \cite{B}.

\smallskip

In a cosmological model Einstein's equations connect the geometry of
the space-time with the properties of the matter. The matter
occupying the space-time is determined by the stress energy tensor
of the matter. In our study we follow \cite{B} and we consider the
hydrodynamical tensor of the matter. We will work with an equation
of state of matter of the form $p=k\varepsilon$,  where
$\varepsilon$ is the energy density of the matter, $p$ is the
pressure and $0\leq k \leq 1$, see also \cite{B}.  The case $k=1$ is
studied in \cite{FLP}. Here we consider the case where $0<k<1$.

\smallskip

Following \cite{B} the Einstein equations for the homogenous
cosmologies of class A without motion of matter can be formalized as
a Hamiltonian system in the phase space $p_i,q_i$ for $i=1,2,3$ with
the Hamiltonian function
\[
H=\dfrac{1}{(q_1q_2q_3)^{\frac{1-k}{2}}}\left(T(p_iq_i)+\dfrac{1}{4}
V_G(q_i)\right).
\]
Here $T$ is the kinetic energy (not positive defined) and $V_G$ is
the potential. According to  \cite{B} (Section 4 of Chapter II) the
kinetic and the potential energy are given, respectively, by
\[
\begin{array}{ccc}
T(p_iq_i)&=&2\sum\limits_{i<j}^3
p_ip_jq_iq_j-\sum\limits_{i=1}^3p_i^2q_i^2, \\
V_G(q_i)&=&2\sum\limits_{i<j}^3n_in_jq_iq_j-\sum\limits_{i=1}^3n_i^2q_i^2,\\
\end{array}
\]
for $i,j\in\{1,2,3\}$.

We consider the Hamiltonian system
\[
\dot{q}_i=\dfrac{\partial H}{\partial p_i},\qquad \dot{p}_i=-\dfrac{
\partial H}{\partial q_i},
\]
where the dot denotes derivative with respect to the time $t$. More
precisely, the Hamiltonian system is \footnotesize
\[
\begin{split}
\dot{q}_1=&2q_1(q_1q_2q_3)^{\frac{k-1}{2}}(-p_1q_1+p_2q_2+p_3q_3),\\
\dot{q}_2=&2q_2(q_1q_2q_3)^{\frac{k-1}{2}}(p_1q_1-p_2q_2+p_3q_3),\\
\dot{q}_3=&2q_3(q_1q_2q_3)^{\frac{k-1}{2}}(p_1q_1+p_2q_2-p_3q_3),\\
\dot{p}_1=&-(q_1q_2q_3)^{\frac{k-1}{2}}\left(2p_1(-p_1q_1+p_2q_2+p_3q_3)
+\dfrac{1}{2}n_1(-n_1q_1+n_2q_2+n_3q_3)\right)+\dfrac{1-k}{2q_1}H,\\
\dot{p}_2=&-(q_1q_2q_3)^{\frac{k-1}{2}}\left(2p_2(p_1q_1-p_2q_2+p_3q_3)
+\dfrac{1}{2}n_2(n_1q_1-n_2q_2+n_3q_3)\right)+\dfrac{1-k}{2q_2}H,\\
\dot{p}_3=&-(q_1q_2q_3)^{\frac{k-1}{2}}\left(2p_3(p_1q_1+p_2q_2-p_3q_3)
+\dfrac{1}{2}n_3(n_1q_1+n_2q_2-n_3q_3)\right)+\dfrac{1-k}{2q_3}H.
\end{split}
\]\normalsize

\begin{table}[b]
\begin{center}
\begin{tabular}{|c||c|c|c|c|c|c|}
\hline
Type & I & II & VI$_0$ & VII$_0$ & VIII & IX\cr\hline\hline
$a$ &  0   & 0  & 0    & 0     & 0    & 0 \cr\hline
$n_1$& 0   & 1  & 1    & 1    & 1     & 1 \cr\hline
$n_2$& 0   & 0  & $-1$   & 1     & 1    & 1 \cr\hline
$n_3$& 0   & 0  & 0     & 0     & $-1$  & 1 \cr\hline
\end{tabular}
\end{center}
\smallskip
\caption{The classification of Bianchi class A cosmologies.
}\label{BianchiClassA}
\end{table}

Note that the constants $n_1,n_2,n_3$ determine the type of the
model according to Table \ref{BianchiClassA}. After the change of
coordinates $ds=(q_1q_2q_3)^{\frac{1-k}{2}}dt$, $q_i=x_i$,
$p_i=x_{i+3}/(2x_i)$, $i=1,2,3$, we obtain the quadratic homogeneous
polynomial differential system
\begin{equation} \label{BB}
\begin{array}{ccl}
\dot{x}_1&=&x_1(-x_4+x_5+x_6),\\
\dot{x}_2&=&x_2 (x_4 - x_5 + x_6),\\
\dot{x}_3&=&x_3 (x_4 + x_5 - x_6),\\
\dot{x}_4&=&n_1x_1(n_1x_1-n_2x_2-n_3x_3)+\frac{k-1}{4}F,\\
\dot{x}_5&=&n_2x_2(-n_1x_1+n_2x_2-n_3x_3)+\frac{k-1}{4}F,\\
\dot{x}_6&=&n_3x_3(-n_1x_1-n_2x_2+n_3x_3)+\frac{k-1}{4}F,
\end{array}
\end{equation}
where
\begin{equation}\label{F}
\begin{split}
F=&n_1^2x_1^2+n_2^2x_2^2+n_3^2x_3^2-2n_1n_2x_1x_2-2n_1n_3x_1x_3-2n_2n_3x_2x_3\\
&+x_4^2+ x_5^2 + x_6^2 - 2 x_4 x_5   - 2x_5x_6-2x_4x_6.
\end{split}
\end{equation}

Note that system \eqref{BB} is a homogeneous polynomial differential
system of degree 2. The Hamiltonian $H$ becomes after the changes of variables the first integral
\begin{equation}\label{Hx}
\begin{split}
\mathcal H=&(x_1x_2x_3)^{\frac{k-1}2}F\\
=&(x_1x_2x_3)^{\frac{k-1}2}(n_1^2x_1^2+n_2^2x_2^2+n_3^2x_3^2-2n_1n_2x_1x_2-2n_1n_3x_1x_3\\
&-2n_2n_3x_2x_3+x_4^2+ x_5^2 + x_6^2 - 2 x_4 x_5-2x_5x_6-2x_4x_6)
\end{split}
\end{equation}
of system \eqref{BB}.

\smallskip

Let $U$ be an open and dense subset of $\R^6$. Then we recall that
system \eqref{BB} has a first integral $\mathcal H: U\to \R$ if
$\mathcal H$ is a non-constant $\mathcal C^1$-function such that
\[
\dot x_1\pd{\mathcal H}{x_1}+\cdots+\dot x_6\pd{\mathcal H}{x_6}=0.
\]

Many authors have studied some models of Class A for the case $k=1$
considering different types of integrability, see for exemple [4--8,
10--15]. In this work we study the analytic integrability of all
Bianchi models of class A in the variables
$(x_1,x_2,x_3,x_4,x_5,x_6)$ for $0\leq k<1$. The following result is
well known, see for instance \cite{LZ}.

\begin{proposition}\label{p1}
Let $F$ be an analytic function and let $F= \sum_{i} F_{i}$ be its
decomposition into homogeneous polynomials of degree $i$. Then $F$
is an analytic first integral of the homogeneous differential system
\eqref{BB} if and only if $F_{i}$ is a homogeneous polynomial first
integral of system \eqref{BB} for all $i$.
\end{proposition}

A differential system of $n$ variables is {\it completely
integrable} if it admits $n-1$ independent first integrals.

\smallskip

According to Proposition \ref{p1} the study of the analytic first
integrals of the homogeneous system \eqref{BB} is reduced to the
study of its polynomial homogeneous first integrals. The  main
result of this paper is the characterization of the polynomial first
integrals of the Bianchi models of class A. Section \ref{S:lemes}
provides three technical lemmas that we will use in Section
\ref{S:prova} to prove the following theorem.

\begin{theorem}\label{main}
For $0\leq k<1$ the following statements hold.
\begin{itemize}
\item[(a)] The Bianchi type {\em I} model is completely integrable.

\item[(b)]  The Bianchi type {\em II} model has the polynomial first
integral $K=x_5-x_6$. This model does not admit any additional
polynomial first integral independent from $\mathcal H$ and $K$.

\item[(c)] The Bianchi type {\em VI}$_0$ and  {\em VII}$_0$ models
have no polynomial first integrals.

\item[(d)] The Bianchi type  {\em VIII} and  {\em IX} models have
no polynomial first integrals.
\end{itemize}
\end{theorem}

Statement (a) of Theorem \ref{main} apparently is well known for
people working in the area, but we cannot find a reference where
 it is proved.

\section{Some auxiliary lemmas}\label{S:lemes}

In order to prove Theorem \ref{main} we shall use the following
three lemmas.

\begin{lemma}[see \cite{LV3}]\label{L7}
Let $x_k$ be a one-dimensional variable, $k\in\{1,\ldots,n\}$, $n>1$
and let $f=f(x_1,\ldots,x_n)$ be a polynomial. For
$l\in\{1,\cdots,n\}$ and $c_0$ a constant let
$f_l=f(x_1,\ldots,x_n)|_{x_l=c_0}$. Then there exists a polynomial
$g=g(x_1,\ldots,x_n)$ such that $f= f_l+(x_l-c_0)g$.
\end{lemma}

\begin{lemma}\label{L:Eq.estrella}
Let $g=g(x_4,x_5,x_6)$ be a homogeneous polynomial solution of the
homogeneous partial differential equation
\begin{equation}\label{Eq:estrella}
(a_1x_4+a_2x_5+a_3x_6)g+\frac{k-1}4F_{123}\left(\pd g{x_4}+\pd
g{x_5}+\pd g{x_6}\right)=0,
\end{equation}
where $F_{123}=x_4^2+x_5^2+x_6^2-2(x_4x_5+x_4x_6+x_5x_6)$ and
$a_1,a_2,a_3\in\R$
 are such that $(a_1-a_2)^2+(a_1-a_3)^2\neq0$. Then
$g\equiv0$.
\end{lemma}

\begin{proof}
The general solution of equation \eqref{Eq:estrella} is\small
\[
\begin{split}
g(x_4,x_5,x_6)=&f(x_4-x_5,x_4-x_6)\\
&\left(-x_4-x_5-x_6+2\sqrt{\Delta}\right)^{\Delta_1+\Delta_2}
\left(x_4+x_5+x_6+2\sqrt{\Delta}\right)^{\Delta_1-\Delta_2},
\end{split}
\]\normalsize
where $\Delta_1=2(a_1+a_2+a_3)/(3(k-1))$,
$\Delta_2=((2a_1-a_2-a_3)x_4+(-a_1+2a_2-a_3)x_5+(-a_1-a_2+2a_3)x_6)
/(3(k-1)\sqrt{\Delta})$, $\Delta=x_4^2+x_5^2+x_6^2-x_4x_5-
x_4x_6-x_5x_6$ and $f$ is an arbitrary function. We note that $g$ is
a polynomial if and only if $\Delta_1\in\N$, $\Delta_2=0$ and $f$ is
a polynomial. In particular, the relation $\Delta_2=0$ is equivalent
to the linear  system
\[
\left(\begin{array}{ccc}
      2&-1&-1\cr
-1&2&-1\cr
-1&-1&2
      \end{array}
\right)\left(\begin{array}{c}a_1\cr a_2\cr
a_3\end{array}\right)=\left(\begin{array}{c}0\cr 0\cr
0\end{array}\right).
\]
The solution of this  system  is $a_1=a_2=a_3$. This cannot happen
by assumption. Therefore $g$ is not a polynomial unless $g\equiv 0$.
\end{proof}

\begin{lemma}\label{L:Eq.dificil}
Let $g=g(x_4,x_5,x_6)$ and $h=h(x_4-x_5,x_4-x_6)$ be homogeneous
polynomials of respective degrees $n-2$ and $n$ such that
\begin{equation}\label{Eq:dificil}
2(x_4-x_5+x_6)g+\frac{k-1}4F_{123}\left(\pd g{x_4}+\pd g{x_5}+\pd
g{x_6}\right)+\pd{h}{x_5}=0,
\end{equation}
where $F_{123}=x_4^2+x_5^2+x_6^2-2(x_4x_5+x_4x_6+x_5x_6)$. Then
$h=h(x_4-x_6)$ and $g\equiv0$.
\end{lemma}

\begin{proof}
Let $h=\sum\limits_{i=0}^na_i(x_4-x_5)^i(x_4-x_6)^{n-i}$ and
$g=\sum\limits_{i+j=0}^{n-2}b_{ij}x_4^ix_5^jx_6^{n-2-i-j}$. Suppose
that $g\not\equiv0$. Forcing that the solution of \eqref{Eq:dificil}
be a polynomial, Mathematica (see \cite{M}) shows that $g$ is of the
form
\[
g(x_4,x_5,x_6)=\frac{4}{1-k}\int\frac{{h}_5}{F_{123}}\,dx_4+f(x_4-x_5,x_4-x_6),
\]
where $f$ is a homogeneous polynomial, ${h}_5=\pd{h}{x_5}$ and the integral is to be a polynomial.

\smallskip

Let $A_1=\sqrt{x_5}-\sqrt{x_6}$ and $A_2=\sqrt{x_5}+\sqrt{x_6}$. Under this notation $F_{123}=(x_4-A_1^2)(x_4-A_2^2)$. The fraction inside the above integral can be written as
\[
\frac{{h}_5}{F_{123}}=X_0+\frac 1{A_1^2-A_2^2}\left(\frac{X_1}{x_4-A_1^2}-\frac{X_2}{x_4-A_2^2}\right),
\]
where $X_0=X_0(x_4,A_1,A_2)$, $X_1=X_1(A_1,A_2)$ and $X_2=X_2(A_1,A_2)$ are homogeneous polynomials. The integrals of the fractions in the right hand side with respect to $x_4$ are $X_i\log(x_4-A_i^2)$, $i=1,2$; hence $X_1$ and $X_2$ must be identically zero. $X_1=0$ and $X_2=0$ have the same solutions $a_1,\ldots,a_n$ because of symmetry. Indeed $X_1=0$ (or $X_2=0$) reduces to $S_n=0$, where
\[
S_n=\sum_{i=1}^n(3A_1-A_2)^{n-i}(A_1+A_2)^{n-i}(3A_1+A_2)^{i-1}(A_1-A_2)^{i-1}i\,a_i.
\]
We note that we have the recursive equality
\[
S_n=(3A_1-A_2)(A_1+A_2)S_{n-1}+(3A_1+A_2)^{n-1}(A_1-A_2)^{n-1}n\,a_n.
\]
On $A_1=-A_2$ (or equivalently on $x_5=0$) we have $n4^{n-1}A_2^{2n-2}a_n=0$, and hence we have $a_n=0$. Induction arguments prove that $S_n=0$ implies $a_1=\cdots=a_n=0$. Therefore $h_5=0$, which means that equation \eqref{Eq:dificil} is a particular case of equation \eqref{Eq:estrella} and then by Lemma \ref{L:Eq.estrella} we get $g\equiv0$ and then we are finished.
\end{proof}

\section{Proof of Theorem \ref{main}}\label{S:prova}

In this section we prove the four statements of Theorem \ref{main}.

\subsection{Proof of statement (a) of Theorem \ref{main}}

According to Table \ref{BianchiClassA} the Bianchi cosmological
model I is obtained for $n_1=n_2=n_3=0$. System \eqref{BB} becomes
\begin{equation} \label{BI}
\begin{split}
\dot{x}_1=&x_1(-x_4+x_5+x_6),\\
\dot{x}_2=&x_2 (x_4 - x_5 + x_6),\\
\dot{x}_3=&x_3 (x_4 + x_5 - x_6),\\
\dot{x}_4=&\dfrac{k-1}{4}F,\\
\dot{x}_5=&\dfrac{k-1}{4}F,\\
\dot{x}_6=&\dfrac{k-1}{4}F,
\end{split}
\end{equation}
where $F=(x_4^2 + x_5^2 + x_6^2 - 2 x_4 x_5 - 2x_5x_6-2x_4x_6)$.
Straightforward computations show that system \eqref{BI} has the
five first integrals $x_4-x_5$, $x_4-x_6$, $\mathcal H$ defined in \eqref{Hx},
\[
\left(\frac{x_1}{x_2}\right)^{\frac{1-k}2}\left(\frac{x_4+x_5+x_6-
2\sqrt{\Delta}}{x_4+x_5+x_6+2\sqrt{\Delta}}\right)_,^{\frac{x_4-x_5}
{\sqrt{\Delta}}}
\]
and
\[
\left(\frac{x_2}{x_3}\right)^{\frac{1-k}2}\left(\frac{x_4+x_5+x_6-
2\sqrt{\Delta}}{x_4+x_5+x_6+2\sqrt{\Delta}}\right)_,^{\frac{x_5-x_6}
{\sqrt{\Delta}}},
\]
with $\Delta=x_4^2+x_5^2+x_6^2-x_4x_5-x_5x_6-x_4x_6$. Note that the
five first integrals are independent. Statement (a) is proved.

\subsection{Prove of statement (b) of Theorem \ref{main}}

The Bianchi cosmological model II is obtained for $n_1=1$ and
$n_2=n_3=0$. System \eqref{BB} writes as
\begin{equation} \label{BII}
\begin{split}
\dot{x}_1=&x_1(-x_4+x_5+x_6),\cr
\dot{x}_2=&x_2 (x_4 - x_5 + x_6),\cr
\dot{x}_3=&x_3 (x_4 + x_5 - x_6),\cr
\dot{x}_4=&x_1^2+\frac{k-1}{4}F,\cr
\dot{x}_5=&\frac{k-1}{4}F,\cr
\dot{x}_6=&\frac{k-1}{4}F,
\end{split}
\end{equation}
where $F=x_1^2+x_4^2  + x_5^2 + x_6^2 - 2 x_4 x_5 -
2x_5x_6-2x_4x_6$.

\smallskip

Let $h=h(x_1,x_2,x_3,x_4,x_5,x_6)$ be a homogeneous polynomial first
integral of \eqref{BII}. Using Lemma \ref{L7} we can write
$h=h_1(x_2,x_3,x_4,x_5,x_6)+x_1^jg_1(x_1,x_2,x_3,x_4,x_5,x_6)$,
with $j\in\N$ and $h_1$ and $g_1$ homogeneous polynomials such that
$x_1\nmid g_1$. On $x_1=0$ system \eqref{BII} becomes
\begin{equation}\label{BIIx1}
\begin{split}
\dot{x}_2=&x_2 (x_4 - x_5 + x_6),\cr
\dot{x}_3=&x_3 (x_4 + x_5 - x_6),\cr
\dot{x}_4=&\frac{k-1}{4}F_1,\cr
\dot{x}_5=&\frac{k-1}{4}F_1,\cr
\dot{x}_6=&\frac{k-1}{4}F_1
\end{split}
\end{equation}
where $F_1=F|_{x_1=0}$. System \eqref{BIIx1} admits the two
polynomial first integrals $x_4-x_5$ and $x_5-x_6$ and the two
non-polynomial first integrals
\[
x_2^{\frac 32
(k-1)}F_1\left(\frac{x_4+x_5+x_6-2\sqrt{\Delta}}{x_4+x_5+x_6+
2\sqrt{\Delta}}\right)^{\frac{x_4-2x_5+x_6}{\sqrt{\Delta}}}
\]
and
\[
x_3^{\frac 32
(k-1)}F_1\left(\frac{x_4+x_5+x_6-2\sqrt{\Delta}}{x_4+x_5+x_6+
2\sqrt{\Delta}}\right)^{\frac{x_4+x_5-2x_6}{\sqrt{\Delta}}},
\]
where $\Delta=x_4^2+x_5^2+x_6^2-x_4x_5-x_4x_6-x_5x_6$. As these four
first integrals of system \eqref{BIIx1} are independent and $h_1$ is
a polynomial first integral of \eqref{BIIx1}, we have
$h_1=h_1(x_4-x_5,x_5-x_6)$.

\smallskip

The following lemma ends the proof of statement (b) of Theorem \ref{main}.

\begin{lemma}\label{L:g1}
For system \eqref{BII} we have that $h_1=h_1(x_5-x_6)$ and $g_1\equiv0$.
\end{lemma}

\begin{proof}
Suppose that $g_1\not\equiv0$. As $h$ is a first integral of \eqref{BII}, we have
\begin{equation}\label{g1}
\begin{split}
x_1^j&\left[j(-x_4+x_5+x_6)g_1+x_1(-x_4+x_5+x_6)\pd{g_1}{x_1}+x_2 (x_4 - x_5 + x_6)\pd{g_1}{x_2}\right.\\
&\left.+x_3 (x_4 + x_5 - x_6)\pd{g_1}{x_3}+x_1^2\pd{g_1}{x_4}+\frac{k-1}{4}F\left(\pd{g_1}{x_4}+\pd{g_1}{x_5}+\pd{g_1}{x_6}\right)\right]+x_1^2\pd{h_1}{x_4}=0.
\end{split}
\end{equation}

We distinguish three cases depending on the value of $j$. If $j=1$ then equation \eqref{g1} becomes
\[
\begin{split}
(-x_4&+x_5+x_6)g_1+x_1(-x_4+x_5+x_6)\pd{g_1}{x_1}+x_2 (x_4 - x_5 + x_6)\pd{g_1}{x_2}\\
&+x_3 (x_4 + x_5 - x_6)\pd{g_1}{x_3}+x_1^2\pd{g_1}{x_4}+\frac{k-1}{4}F\left(\pd{g_1}{x_4}+\pd{g_1}{x_5}+\pd{g_1}{x_6}\right)+x_1\pd{h_1}{x_4}=0.
\end{split}
\]

Let $\bar g_1=g_1|_{x_1=0}\not\equiv0$. Equation \eqref{g1} on $x_1=0$ can be written as
\[
\begin{split}
&(-x_4+x_5+x_6)\bar g_1+x_2 (x_4 - x_5 + x_6)\pd{\bar g_1}{x_2}\\
&+x_3 (x_4 + x_5 - x_6)\pd{\bar g_1}{x_3}+\frac{k-1}{4}F_1\left(
\pd{\bar g_1}{x_4}+\pd{\bar g_1}{x_5}+\pd{\bar g_1}{x_6}\right)=0.
\end{split}
\]
Write $\bar g_1=x_2^lg_2\not\equiv0$, with $l\in\N\cup\{0\}$ and $x_2\nmid g_2$. We get
\[
\begin{split}
((-x_4&+x_5+x_6)+l(x_4 - x_5 + x_6))g_2+x_2 (x_4 - x_5 + x_6)\pd{g_2}{x_2}\\
&+x_3 (x_4 + x_5 - x_6)\pd{g_2}{x_3}+ \frac{k-1}{4}F_1\left(\pd{g_2}{x_4}+ \pd{g_2}{x_5}+\pd{g_2}{x_6}\right)=0.
\end{split}
\]
Let $\bar g_2=g_2|_{x_2=0}\not\equiv0$. On $x_2=0$ we have
\[
\begin{split}
((-x_4&+x_5+x_6)+l(x_4 - x_5 + x_6))\bar g_2\\
&+x_3 (x_4 + x_5 -x_6)\pd{\bar g_2}{x_3}+\frac{k-1}{4}F_{12}\left(\pd{\bar g_2}{x_4}+\pd{\bar g_2}{x_5}+\pd{\bar g_2}{x_6}\right)=0,
\end{split}
\]
where $F_{12}=F_1|_{x_2=0}$. Now we write $\bar
g_2=x_3^mg_3\not\equiv0$, with $m\in\N\cup\{0\}$ and $x_3\nmid g_3$.
We obtain
\[
\begin{split}
((-x_4&+x_5+x_6)+l(x_4 - x_5 + x_6)+m(x_4 + x_5 - x_6))g_3\\
&+x_3 (x_4 + x_5 - x_6)\pd{g_3}{x_3}+\frac{k-1}{4}F_{12}\left(\pd{g_3}{x_4}+\pd{g_3}{x_5}+\pd{g_3}{x_6}\right)=0.
\end{split}
\]
Let $\bar g_3=g_3|_{x_3=0}\not\equiv0$. On $x_3=0$ we have
\[
\begin{split}
((-x_4&+x_5+x_6)+l(x_4 - x_5 + x_6)+m(x_4 + x_5 - x_6))\bar g_3\\
&+\frac{k-1}{4}F_{123}\left(\pd{\bar g_3}{x_4}+\pd{\bar g_3}{x_5}+\pd{\bar g_3}{x_6}\right)=0,
\end{split}
\]
where $F_{123}=F_{12}|_{x_3=0}$. Applying Lemma \ref{L:Eq.estrella} we get $\bar g_3\equiv0$, which is a contradiction. Hence we have
$g_1\equiv0$ and therefore $\pd{h_1}{x_4}\equiv0$. The lemma follows in this case.

\smallskip

If $j>2$ then $x_1\left|\pd{h_1}{x_4}\right.$, and then $\pd{h_1}{x_4}\equiv0$. Now we can proceed in a similar way as in the case $j=1$ to prove that $g_1\equiv0$ by using Lemma \ref{L:Eq.estrella}.

\smallskip

If $j=2$ then equation \eqref{g1} becomes
\[
\begin{split}
2(-x_4&+x_5+x_6)g_1+x_1(-x_4+x_5+x_6)\pd{g_1}{x_1}+x_2 (x_4 - x_5 + x_6)\pd{g_1}{x_2}\\
&+x_3 (x_4 + x_5 - x_6)\pd{g_1}{x_3}+x_1^2\pd{g_1}{x_4}+\frac{k-1}{4}F\left(\pd{g_1}{x_4}+\pd{g_1}{x_5}+\pd{g_1}{x_6}\right)+\pd{h_1}{x_4}=0.
\end{split}
\]

Let $\bar g_1=g_1|_{x_1=0}\not\equiv0$. Equation \eqref{g1} on $x_1=0$ can be written as
\[
\begin{split}
&2(-x_4+x_5+x_6)\bar g_1+x_2 (x_4 - x_5 + x_6)\pd{\bar g_1}{x_2}\\
&+x_3 (x_4 + x_5 - x_6)\pd{\bar g_1}{x_3}+\frac{k-1}{4}F_1\left(\pd{\bar g_1}{x_4}+\pd{\bar g_1}{x_5}+\pd{\bar g_1}{x_6}\right)+\pd{h_1}{x_4}=0.
\end{split}
\]
Write $\bar g_1=x_2^lg_2\not\equiv0$, with $l\in\N\cup\{0\}$ and $x_2\nmid g_2$. We get
\[
\begin{split}
(2(-x_4&+x_5+x_6)+l(x_4 - x_5 + x_6))g_2+x_2 (x_4 - x_5 + x_6)\pd{g_2}{x_2}\\
&+x_3 (x_4 + x_5 - x_6)\pd{g_2}{x_3}+ \frac{k-1}{4}F_1\left(\pd{g_2}{x_4}+ \pd{g_2}{x_5}+\pd{g_2}{x_6}\right)+\pd{h_1}{x_4}=0.
\end{split}
\]
If $l>0$ then $\pd{h_1}{x_4}\equiv0$. Similar arguments to those used above lead to the desired result after applying Lemma \ref{L:Eq.estrella}. If $l=0$, let $\bar g_2=g_2|_{x_2=0}\not\equiv0$. On $x_2=0$ we have
\[
2(-x_4+x_5+x_6)\bar g_2+x_3 (x_4 + x_5 -x_6)\pd{\bar g_2}{x_3}+\frac{k-1}{4}F_{12}\left(\pd{\bar g_2}{x_4}+\pd{\bar g_2}{x_5}+\pd{\bar g_2}{x_6}\right)+\pd{h_1}{x_4}=0,
\]
where $F_{12}=F_1|_{x_2=0}$. Now we write $\bar g_2=x_3^mg_3\not\equiv0$, with $m\in\N\cup\{0\}$ and $x_3\nmid g_3$.
We obtain
\[
\begin{split}
(2(-x_4&+x_5+x_6)+m(x_4 + x_5 - x_6))g_3\\
&+x_3 (x_4 + x_5 - x_6)\pd{g_3}{x_3}+\frac{k-1}{4}F_{12}\left(\pd{g_3}{x_4}+\pd{g_3}{x_5}+\pd{g_3}{x_6}\right)+\pd{h_1}{x_4}=0.
\end{split}
\]
If $m>0$ then $\pd{h_1}{x_4}\equiv0$. Again the usual arguments lead to the desired result after applying Lemma \ref{L:Eq.estrella}. If $m=0$, let $\bar g_3=g_3|_{x_3=0}\not\equiv0$. On $x_3=0$ we have
\[
2(-x_4+x_5+x_6)\bar g_3+\frac{k-1}{4}F_{123}\left(\pd{\bar g_3}{x_4}+\pd{\bar g_3}{x_5}+\pd{\bar g_3}{x_6}\right)+\pd{h_1}{x_4}=0,
\]
where $F_{123}=F_{12}|_{x_3=0}$. Applying Lemma \ref{L:Eq.dificil} swapping $x_4$ and $x_5$ we get $\bar g_3\equiv0$, which is a contradiction. Hence we have
$g_1\equiv0$ and therefore $\pd{h_1}{x_4}\equiv0$. The lemma follows also in this case.

\end{proof}

After Lemma \ref{L:g1}, $h=h(x_5-x_6)$. Hence statement (b) of Theorem \ref{main} follows.

\subsection{Proof of statement (c) of Theorem \ref{main}}

According to Table \ref{BianchiClassA}, system \eqref{BB} in cases
VI$_0$ and VII$_0$  can be written as
\begin{equation} \label{B67}
\begin{split}
\dot{x}_1&=x_1(-x_4+x_5+x_6),\\
\dot{x}_2&=x_2 (x_4 - x_5 + x_6),\\
\dot{x}_3&=x_3 (x_4 + x_5 - x_6),\\
\dot{x}_4&=x_1(x_1-n_2x_2)+\frac{k-1}{4}F,\\
\dot{x}_5&=n_2x_2(-x_1+n_2x_2)+\frac{k-1}{4}F,\\
\dot{x}_6&=\frac{k-1}{4}F,
\end{split}
\end{equation}
where $F=(x_1-n_2x_2)^2+x_4^2+ x_5^2 + x_6^2 - 2 x_4 x_5   -2x_5x_6-2x_4x_6$ and $n_2^2=1$. Suppose that system \eqref{B67} has a homogeneous
polynomial first integral $h(x_1,\ldots,x_6)$. From Lemma \ref{L7}
we can write $h=h_1(x_2,\ldots,x_6)+x_1^jg_1(x_1,\ldots,x_6)$, with
$j\in\N$ and $h_1$ and $g_1$ homogeneous polynomials such that
$x_1\nmid g_1$. System \eqref{B67} on $x_1=0$ is
\begin{equation} \label{B67_1}
\begin{split}
\dot{x}_2&=x_2 (x_4 - x_5 + x_6),\\
\dot{x}_3&=x_3 (x_4 + x_5 - x_6),\\
\dot{x}_4&=\frac{k-1}{4}F_1,\\
\dot{x}_5&=x_2^2+\frac{k-1}{4}F_1,\\
\dot{x}_6&=\frac{k-1}{4}F_1,
\end{split}
\end{equation}
where $F_1=F|_{x_1=0}$. We note that $h_1$ is a first integral of
system \eqref{B67_1}. From Lemma \ref{L7} we can write
$h_1=h_2(x_3,\ldots,x_6)+x_2^lg_2(x_2,\ldots,x_6)$, with $l\in\N$
and $h_2$ and $g_2$ homogeneous polynomials such that $x_2\nmid
g_2$. System \eqref{B67_1} on $x_2=0$ writes
\begin{equation} \label{B67_12}
\begin{split}
\dot{x}_3&=x_3 (x_4 + x_5 - x_6),\\
\dot{x}_4&=\frac{k-1}{4}F_{12},\\
\dot{x}_5&=\frac{k-1}{4}F_{12},\\
\dot{x}_6&=\frac{k-1}{4}F_{12},
\end{split}
\end{equation}
where $F_{12}=F_1|_{x_2=0}$. We note that $h_2$ is a first integral
of system \eqref{B67_12}. Straightforward computations show that
system \eqref{B67_12} has the three independent first integrals
$x_4-x_5$, $x_5-x_6$ and
\[
x_3^{\frac 32 (k-1)}F_{12}\left(\frac{x_4+x_5+ x_6- 2\sqrt{\Delta}}
{x_4+x_5+x_6+2\sqrt{\Delta}}\right)_,^{\frac{x_4+x_5-2x_6}
{\sqrt{\Delta}}},
\]
where $\Delta=x_4^2+x_5^2+x_6^2-x_4x_5-x_4x_6-x_5x_6$. Therefore
$h_2=h_2(x_4-x_5,x_4-x_6)$.

\begin{lemma}\label{L:g2h2}
For system \eqref{B67_1} we have that $h_2=h_2(x_4-x_6)$ and
$g_2\equiv0$.
\end{lemma}

\begin{proof}
Suppose that $g_2\not\equiv0$. As $h_1=h_2+x_2^lg_2$ is a first
integral of system \eqref{B67_1}, we can write
\begin{equation}\label{Eq:VIx2}
\begin{split}
x_2^l&\left[l(x_4-x_5+x_6)g_2+x_2(x_4-x_5+x_6)\pd{g_2}{x_2}+
x_3(x_4+x_5-x_6)\pd{g_2}{x_3}\right.\\
&\left.+x_2^2\pd{g_2}{x_5}+\frac{k-1}4F_1\left(\pd{g_2}{x_4}+
\pd{g_2}{x_5}+\pd{g_2}{x_6}\right)\right]+x_2^2\pd{h_2}{x_5}=0.
\end{split}
\end{equation}
We distinguish three cases depending on the value of $l$. If $l=1$
then equation \eqref{Eq:VIx2} becomes
\[
\begin{split}
(x_4&-x_5+x_6)g_2+x_2(x_4-x_5+x_6)\pd{g_2}{x_2}+x_3(x_4+x_5-x_6)
\pd{g_2}{x_3}\\
&+\frac{k-1}4F_1\left(\pd{g_2}{x_4}+\pd{g_2}{x_5}+\pd{g_2}{x_6}\right)
+x_2\pd{h_2}{x_5}+x_2^2\pd{g_2}{x_5}=0.
\end{split}
\]
Let $\bar g_2=g_2|_{x_2=0}\not\equiv0$. On $x_2=0$ we have
\[
(x_4-x_5+x_6)\bar g_2+x_3(x_4+x_5-x_6)\pd{\bar g_2}{x_3}+\frac{k-1}4F_{12}\left(\pd{\bar g_2}{x_4}+\pd{\bar g_2}{x_5}+\pd{\bar g_2}{x_6}\right)=0.
\]
Write $\bar g_2=x_3^mg_3\not\equiv0$, with $m\in\N\cup\{0\}$ and
$x_3\nmid g_3$. Then
\[
\begin{split}
((x_4&-x_5+x_6)+m(x_4+x_5-x_6))g_3+x_3(x_4+x_5-x_6)\pd{g_3}{x_3}\\
&+\frac{k-1}4F_{12}\left(\pd{g_3}{x_4}+\pd{g_3}{x_5}+
\pd{g_3}{x_6}\right)=0.
\end{split}
\]
Let $\bar g_3=g_3|_{x_3=0}\not\equiv0$. On $x_3=0$ we get
\[
((x_4-x_5+x_6)+m(x_4+x_5-x_6))\bar
g_3+\frac{k-1}4F_{123}\left(\pd{\bar g_3}{x_4}+\pd{\bar
g_3}{x_5}+\pd{\bar g_3}{x_6}\right)=0,
\]
where $F_{123}=F_{12}|_{x_3=0}$. Applying Lemma \ref{L:Eq.estrella}
we obtain $\bar g_3\equiv0$, which is a contradiction. Hence
$g_2\equiv0$. Back to equation \eqref{Eq:VIx2} we have
$\pd{h_2}{x_5}\equiv0$. Then the lemma follows.

\smallskip

If $l>2$, then from equation \eqref{Eq:VIx2} we have that
$x_2\left|\pd{h_2}{x_5}\right.$ and thus $\pd{h_2}{x_5}\equiv0$.
Therefore $h_2=h_2(x_4-x_6)$. Now we can proceed in a similar way as
in the case $l=1$ to obtain the equation
\[
(l(x_4-x_5+x_6)+m(x_4+x_5-x_6))\bar
g_3+\frac{k-1}4F_{123}\left(\pd{\bar g_3}{x_4}+\pd{\bar
g_3}{x_5}+\pd{\bar g_3}{x_6}\right)=0.
\]
Applying again Lemma \ref{L:Eq.estrella} we arrive to contradiction
and hence $g_2\equiv0$.

\smallskip

If $l=2$, then equation \eqref{Eq:VIx2} writes as
\[
\begin{split}
2(x_4&-x_5+x_6)g_2+x_2(x_4-x_5+x_6)\pd{g_2}{x_2}+x_3(x_4+
x_5-x_6)\pd{g_2}{x_3}\\
&+\frac{k-1}4F_1\left(\pd{g_2}{x_4}+\pd{g_2}{x_5}+\pd{g_2}
{x_6}\right)+x_2^2\pd{g_2}{x_5}+\pd{h_2}{x_5}=0.
\end{split}
\]

Let $\bar g_2=g_2|_{x_2=0}\not\equiv0$. On $x_2=0$ we have
\[
2(x_4-x_5+x_6)\bar g_2+x_3(x_4+x_5-x_6)\pd{\bar g_2}{x_3}+\frac{k-1}4F_{12}\left(\pd{\bar g_2}{x_4}+\pd{\bar g_2}{x_5}+\pd{\bar g_2}{x_6}\right)+\pd{h_2}{x_5}=0.
\]
Write $\bar g_2=x_3^mg_3\not\equiv0$, with $m\in\N\cup\{0\}$ and
$x_3\nmid g_3$. Then
\[
\begin{split}
x_3^m&\left[(2(x_4-x_5+x_6)+m(x_4+x_5-x_6))g_3+x_3(x_4+x_5-x_6)
\pd{g_3}{x_3}\right.\\
&\left.+\frac{k-1}4F_{12}\left(\pd{g_3}{x_4}+\pd{g_3}{x_5}+\
pd{g_3}{x_6}\right)\right]+\pd{h_2}{x_5}=0.
\end{split}
\]

We distinguish two cases depending on the value of $m$. If $m>0$
then $x_3\left|\pd{h_2}{x_5}\right.$. Hence $\pd{h_2}{x_5}\equiv0$
and $h_2=h_2(x_4-x_6)$. Now let $\bar g_3=g_3|_{x_3=0}\not\equiv0$.
On $x_3=0$ we obtain
\[
(2(x_4-x_5+x_6)+m(x_4+x_5-x_6))\bar
g_3+\frac{k-1}4F_{123}\left(\pd{\bar g_3}{x_4}+\pd{\bar
g_3}{x_5}+\pd{\bar g_3}{x_6}\right)=0.
\]
Applying Lemma \ref{L:Eq.estrella} we get a contradiction and hence
$g_2\equiv0$.

\smallskip

If $m=0$, let $\bar g_3=g_3|_{x_3=0}\not\equiv0$. On $x_3=0$ we
obtain
\[
2(x_4-x_5+x_6)\bar g_3+\frac{k-1}4F_{123}\left(\pd{\bar
g_3}{x_4}+\pd{\bar g_3}{x_5}+\pd{\bar
g_3}{x_6}\right)+\pd{h_2}{x_5}=0.
\]
Applying Lemma \ref{L:Eq.dificil} we get $\pd{h_2}{x_5}\equiv0$ and
$\bar g_3\equiv0$, hence the lemma follows.

\smallskip

All the subcases are finished and therefore the lemma is proved.
\end{proof}

After Lemma \ref{L:g2h2} we have that
$h=h_2(x_4-x_6)+x_1^jg_1(x_1,\ldots,x_6)$, with $j\in\N$ and
$x_1\nmid g_1$. We recall that $h$ is a first integral of system
\eqref{B67}. Thus
\begin{equation}\label{Eq:h1}
\begin{split}
x_1^j&\left[j(-x_4+x_5+x_6)g_1+x_1(-x_4+x_5+x_6)\pd{g_1}{x_1}\right.\\
&+x_2(x_4-x_5+x_6)\pd{g_1}{x_2}+x_3(x_4+x_5-x_6)\pd{g_1}{x_3}+x_1
(x_1-n_2x_2)\pd{g_1}{x_4}\\
&\left.-n_2x_2(x_1-n_2x_2)\pd{g_1}{x_5}+\frac{k-1}4F\left(\pd{g_1}{x_4}+
\pd{g_1}{x_5}+\pd{g_1}{x_6}\right)\right]+x_1(x_1-n_2x_2)\pd{h_2}{x_4}=0.
\end{split}
\end{equation}

\begin{lemma}\label{L:g1h2}
For system \eqref{B67} we have that $h_2\equiv0$ and $g_1\equiv0$.
\end{lemma}

\begin{proof}
Suppose that $g_1\not\equiv0$. We distinguish two cases depending on
the value of $j$. If $j>1$ then from equation \eqref{Eq:h1} we have
that $x_1\left|\pd{h_2}{x_4}\right.$, and hence $h_2\equiv0$.
Therefore equation \eqref{Eq:h1} can be written as
\[
\begin{split}
j(-x_4&+x_5+x_6)g_1+x_1(-x_4+x_5+x_6)\pd{g_1}{x_1}+x_2
(x_4-x_5+x_6)\pd{g_1}{x_2}\\
&+x_3(x_4+x_5-x_6)\pd{g_1}{x_3}+x_1(x_1-n_2x_2)\pd{g_1}{x_4}
-n_2x_2(x_1-n_2x_2)\pd{g_1}{x_5}\\
&+\frac{k-1}4F\left(\pd{g_1}{x_4}+\pd{g_1}{x_5}+\pd{g_1}
{x_6}\right)=0.
\end{split}
\]
Let $\bar g_1=g_1|_{x_1=0}\not\equiv0$. Equation \eqref{Eq:h1} on
$x_1=0$ becomes
\[
\begin{split}
j(-x_4&+x_5+x_6)\bar g_1+x_2(x_4-x_5+x_6)\pd{\bar g_1}{x_2}+
x_3(x_4+x_5-x_6)\pd{\bar g_1}{x_3}\\
&+x_2^2\pd{\bar g_1}{x_5}+\frac{k-1}4F_1\left(\pd{\bar g_1}{x_4}+
\pd{\bar g_1}{x_5}+\pd{\bar g_1}{x_6}\right)=0.
\end{split}
\]
Write $\bar g_1=x_2^lg_2\not\equiv0$, with $l\in\N\cup\{0\}$ and
$x_2\nmid g_2$. We get
\[
\begin{split}
(j(-x_4&+x_5+x_6)+l(x_4-x_5+x_6))g_2+x_2(x_4-x_5+x_6)\pd{g_2}{x_2}\\
&+x_3(x_4+x_5-x_6)\pd{g_2}{x_3}+x_2^2\pd{g_2}{x_5}+\frac{k-1}4F_1\left(\pd{g_2}{x_4}+\pd{g_2}
{x_5}+\pd{g_2}{x_6}\right)=0.
\end{split}
\]
Let $\bar g_2=g_2|_{x_2=0}\not\equiv0$. Then, on $x_2=0$ we have
\[
\begin{split}
(j(-x_4&+x_5+x_6)+l(x_4-x_5+x_6))\bar g_2+x_3(x_4+x_5-x_6)\pd{\bar
g_2}{x_3}\\ &+\frac{k-1}4F_{12}\left(\pd{\bar g_2}{x_4}+\pd{\bar
g_2}{x_5}+\pd{\bar g_2}{x_6}\right)=0.
\end{split}
\]
Now write $\bar g_2=x_3^mg_3\not\equiv0$, with $m\in\N\cup\{0\}$ and
$x_3\nmid g_3$. We get
\[
\begin{split}
(j(-x_4&+x_5+x_6)+l(x_4-x_5+x_6)+m(x_4+x_5-x_6))g_3\\
&+x_3(x_4+x_5-x_6)\pd{g_3}{x_3}+\frac{k-1}4F_{12}\left(
\pd{g_3}{x_4}+\pd{g_3}{x_5}+\pd{g_3}{x_6}\right)=0.
\end{split}
\]
Let $\bar g_3=g_3|_{x_3=0}\not\equiv0$. Then, on $x_3=0$ we have
\[
\begin{split}
(j(-x_4&+x_5+x_6)+l(x_4-x_5+x_6)+m(x_4+x_5-x_6))\bar g_3\\
&+\frac{k-1}4F_{123}\left(\pd{\bar g_3}{x_4}+\pd{\bar
g_3}{x_5}+\pd{\bar g_3}{x_6}\right)=0.
\end{split}
\]
Applying Lemma \ref{L:Eq.estrella} we obtain $\bar g_3\equiv0$, a
contradiction. Hence $g_1\equiv0$ and the lemma follows in this
case.

\smallskip

If $j=1$ then equation \eqref{Eq:h1} becomes
\[
\begin{split}
(-x_4&+x_5+x_6)g_1+x_1(-x_4+x_5+x_6)\pd{g_1}{x_1}+x_2
(x_4-x_5+x_6)\pd{g_1}{x_2}\\
&+x_3(x_4+x_5-x_6)\pd{g_1}{x_3}+x_1(x_1-n_2x_2)\pd{g_1}{x_4}\\
&-n_2x_2(x_1-n_2x_2)\pd{g_1}{x_5}+\frac{k-1}4F\left(\pd{g_1}{x_4}
+\pd{g_1}{x_5}+\pd{g_1}{x_6}\right)+(x_1-n_2x_2)\pd{h_2}{x_4}=0.
\end{split}
\]
Let $\bar g_1=g_1|_{x_1=0}\not\equiv0$. On $x_1=0$ we have
\[
\begin{split}
(-x_4&+x_5+x_6)\bar g_1+x_2(x_4-x_5+x_6)\pd{\bar
g_1}{x_2}+x_3(x_4+x_5-x_6)\pd{\bar g_1}{x_3}\\ &+x_2^2\pd{\bar
g_1}{x_5}+\frac{k-1}4F_1\left(\pd{\bar g_1}{x_4}+\pd{\bar
g_1}{x_5}+\pd{\bar g_1}{x_6}\right)-n_2x_2\pd{h_2}{x_4}=0.
\end{split}
\]
Write $\bar g_1=x_2^lg_2\not\equiv0$, with $l\in\N\cup\{0\}$ and
$x_2\nmid g_2$. We get
\begin{equation}\label{Eq:x2}
\begin{split}
x_2^l&\left[((-x_4+x_5+x_6)+l(x_4-x_5+x_6))g_2+x_2(x_4-x_5+x_6)
\pd{g_2}{x_2}\right.\\
&\left.+x_3(x_4+x_5-x_6)\pd{g_2}{x_3}+x_2^2\pd{g_2}{x_5}+\frac{k-1}
 4F_1\left(\pd{g_2}{x_4}+\pd{g_2}{x_5}+\pd{g_2}{x_6}\right)\right]\\
&-n_2x_2\pd{h_2}{x_4}=0.
\end{split}
\end{equation}

We distinguish three cases depending on the value of $l$. If $l>1$
then $x_2\left|\pd{h_2}{x_4}\right.$. Hence $h_2\equiv0$. Thus, from
\eqref{Eq:x2},
\[
\begin{split}
((-x_4&+x_5+x_6)+l(x_4-x_5+x_6))g_2+x_2(x_4-x_5+x_6)\pd{g_2}{x_2}\\
&+x_3(x_4+x_5-x_6)\pd{g_2}{x_3}+x_2^2\pd{g_2}{x_5}+\frac{k-1}4F_1\left(\pd{g_2}{x_4}+\pd{g_2}{x_5}+\pd{g_2}{x_6}\right)=0.
\end{split}
\]
Similar arguments to those used before lead to an equation of type
\eqref{Eq:estrella} and hence applying Lemma \ref{L:Eq.estrella} we
get a contradiction. Therefore $g_1\equiv0$ and the lemma follows.

\smallskip

If $l=0$ then we can use the same arguments to arrive from equation
\eqref{Eq:x2} to an equation of type \eqref{Eq:estrella}, and hence
applying Lemma \ref{L:Eq.estrella} we get a contradiction. Therefore
$g_1\equiv0$ and $h_2\equiv0$, so the lemma follows.

\smallskip

It only remains to consider the case $l=1$. Let $\bar
g_2=g_2|_{x_2=0}\not\equiv0$. From equation \eqref{Eq:x2} on $x_2=0$
we have
\[
2x_6\bar g_2+x_3(x_4+x_5-x_6)\pd{\bar g_2}{x_3}+\frac{k-1}4F_{12}\left(\pd{\bar g_2}{x_4}+\pd{\bar g_2}{x_5}+
\pd{\bar g_2}{x_6}\right)-n_2\pd{h_2}{x_4}=0.
\]
Write $\bar g_2=x_3^mg_3\not\equiv0$, with $m\in\N\cup\{0\}$ and
$x_3\nmid g_3$. We get:
\[
\begin{split}
x_3^m&\left[(2x_6+m(x_4+x_5-x_6))g_3+x_3(x_4+x_5-x_6)
\pd{g_3}{x_3}\right.\\
&\left.+\frac{k-1}4F_{12}\left(\pd{g_3}{x_4}+\pd{g_3}{x_5}+
\pd{g_3}{x_6}\right)\right]-n_2\pd{h_2}{x_4}=0.
\end{split}
\]
If $m>0$ then $x_3\left|\pd{h_2}{x_4}\right.$, and hence
$h_2\equiv0$. Therefore we obtain an equation of type
\eqref{Eq:estrella} and hence by Lemma \ref{L:Eq.estrella} we get
$g_1\equiv0$. If $m=0$, let $\bar g_3=g_3|_{x_3=0}\not\equiv0$. On
$x_3=0$, we have
\begin{equation}\label{Eq:h24}
2x_6\bar g_3+\frac{k-1}4F_{123}\left(\pd{\bar g_3}{x_4}+\pd{\bar
g_3}{x_5}+\pd{\bar g_3}{x_6}\right)-n_2\pd{h_2}{x_4}=0.
\end{equation}
As $h_2=h_2(x_4-x_6)$ is a homogeneous polynomial of degree $n$, we
have $h_2=a_0(x_4-x_6)^n$. Thus $\pd{h_2}{x_4}=a_0n(x_4-x_6)^{n-1}$.
On $x_6=0$ equation \eqref{Eq:h24} writes
\[
\frac{k-1}4(x_4-x_5)^2\left.\left(\pd{\bar g_3}{x_4}+\pd{\bar
g_3}{x_5}+\pd{\bar
g_3}{x_6}\right)\right|_{x_6=0}-n_2a_0nx_4^{n-1}=0.
\]
Therefore we must take $a_0=0$ and hence $h_2\equiv0$. Now equation
\eqref{Eq:h24} is of type \eqref{Eq:estrella} and hence by Lemma
\ref{L:Eq.estrella} we get $g_1\equiv0$ and the lemma follows.
\end{proof}

After Lemma \ref{L:g1h2} statement (c) of Theorem \ref{main} is
proved, as it follows that $h\equiv0$.

\subsection{Proof of statement (d) of Theorem \ref{main}}

According to Table \ref{BianchiClassA}, Bianchi cases VIII and IX
correspond to $n_1=n_2=n_3^2=1$ and can be written into the form
\begin{equation} \label{VIII-IX}
\begin{split}
\dot{x}_1&=x_1(-x_4+x_5+x_6),\\
\dot{x}_2&=x_2 (x_4 - x_5 + x_6),\\
\dot{x}_3&=x_3 (x_4 + x_5 - x_6),\\
\dot{x}_4&=x_1(x_1-x_2-n_3x_3)+\frac{k-1}{4}F,\\
\dot{x}_5&=x_2(-x_1+x_2-n_3x_3)+\frac{k-1}{4}F,\\
\dot{x}_6&=n_3x_3(-x_1-x_2+n_3x_3) +\frac{k-1}{4}F,
\end{split}
\end{equation}
where $F=x_1^2+x_2^2+x_3^2-2x_1x_2-2n_3x_1x_3-2n_3x_2x_3+x_4^2+ x_5^2 + x_6^2 - 2 x_4 x_5   - 2x_5x_6-2x_4x_6$ and $n_3^2=1$. Let $h=h(x_1,\cdots,x_6)$ be a homogeneous polynomial
first integral of degree $n$ of system \eqref{VIII-IX}. Write
$h=h_1(x_2,\cdots,x_6)+x_1^jg_1(x_1,\cdots,x_6)$, with $j\in\N$,
$h_1$ and $g_1$ homogeneous polynomials and $x_1\nmid g_1$. System
\eqref{VIII-IX} on $x_1=0$ becomes
\begin{equation} \label{VIII-IXx1}
\begin{split}
\dot{x}_2&=x_2 (x_4 - x_5 + x_6),\\
\dot{x}_3&=x_3 (x_4 + x_5 - x_6),\\
\dot{x}_4&=\frac{k-1}{4}F_1,\\
\dot{x}_5&=x_2(x_2-n_3x_3)+\frac{k-1}{4}F_1,\\
\dot{x}_6&=-n_3x_3(x_2-n_3x_3) +\frac{k-1}{4}F_1,
\end{split}
\end{equation}
where $F_1=F|_{x_1=0}$. System \eqref{VIII-IXx1} admits
$h_1=h_1(x_2,\cdots,x_6)$ as first integral. Write
$h_1=h_2(x_3,\cdots,x_6)+x_2^lg_2(x_2,\cdots,x_6)$, with $l\in\N$,
$h_2$ and $g_2$ homogeneous polynomials and $x_2\nmid g_2$. System
\eqref{VIII-IXx1} on $x_2=0$ becomes
\begin{equation} \label{VIII-IXx12}
\begin{split}
\dot{x}_3&=x_3 (x_4 + x_5 - x_6),\\
\dot{x}_4&=\frac{k-1}{4}F_{12},\\
\dot{x}_5&=\frac{k-1}{4}F_{12},\\
\dot{x}_6&=x_3^2 +\frac{k-1}{4}F_{12},
\end{split}
\end{equation}
where $F_{12}=F_1|_{x_2=0}$. We note that $h_2=h_2(x_3,\cdots,x_6)$
is a first integral of system \eqref{VIII-IXx12}. Write
$h_2=h_3(x_4,x_5,x_6)+x_3^mg_3(x_3,x_4,x_5,x_6)$, with $m\in\N$,
$h_3$ and $g_3$ homogeneous polynomials and $x_3\nmid g_3$. System
\eqref{VIII-IXx12} on $x_3=0$ is
\begin{equation} \label{VIII-IXx123}
\begin{split}
\dot{x}_4&=\frac{k-1}{4}F_{123},\\
\dot{x}_5&=\frac{k-1}{4}F_{123},\\
\dot{x}_6&=\frac{k-1}{4}F_{123},
\end{split}
\end{equation}
where $F_{123}=F_{12}|_{x_3=0}$. Note that $h_3$ is a polynomial
first integral of system \eqref{VIII-IXx123}. Since system
\eqref{VIII-IXx123} admits the two independent first integrals
$x_4-x_5$ and $x_5-x_6$, any polynomial first integral of
\eqref{VIII-IXx123} must be a polynomial in the variables $x_4-x_5$
and $x_5-x_6$. Therefore $h_3=h_3(x_4-x_5, x_5-x_6)$.

\smallskip

The next three lemmas end the proof of statement (d) of Theorem
\ref{main}. The first one shows that $h_2=h_2(x_4-x_5)$.

\begin{lemma}\label{L:g3h2}
For system \eqref{VIII-IXx12}we have that $g_3\equiv 0$ and
$h_3=h_3(x_4-x_5)$.
\end{lemma}

\begin{proof}
Suppose that $g_3\not\equiv0$. We recall that $h_2=h_3(x_4-x_5,
x_5-x_6)+ x_3^mg_3(x_3,x_4,x_5,x_6)$, where $m\in\N$ and $x_3\nmid
g_3$. As $h_2$ is a first integral of system \eqref{VIII-IXx12}, we
have
\begin{equation}\label{VIII1}
\begin{split}
x_3^m&\left[m(x_4 + x_5 - x_6)g_3+x_3(x_4 + x_5 - x_6)\pd{g_3}{x_3}+
x_3^2\pd{g_3}{x_6}\right.\\
&\left.+\frac{k-1}{4}F_{12}\left(\pd{g_3}{x_4}+\pd{g_3}{x_5}+\pd{g_3}
{x_6}\right)\right]+x_3^2\pd{h_3}{x_6}=0.
\end{split}
\end{equation}
We distinguish three cases depending on the value of $m$.

\smallskip
If $m=1$ then
\[
\begin{split}
(x_4 &+ x_5 - x_6)g_3+x_3(x_4 + x_5 -
x_6)\pd{g_3}{x_3}+x_3^2\pd{g_3}{x_6}\\
&+\frac{k-1}{4}F_{12}\left(\pd{g_3}{x_4}+\pd{g_3}{x_5}+\pd{g_3}{x_6}
\right)+x_3\pd{h_3}{x_6}=0.
\end{split}
\]
Let $\bar g_3=g_3|_{x_3=0}\not\equiv0$. On $x_3=0$ we have
\[
(x_4 + x_5 - x_6)\bar g_3+\frac{k-1}{4}F_{123}\left(\pd{\bar
g_3}{x_4}+\pd{\bar g_3}{x_5}+\pd{\bar g_3}{x_6}\right)=0.
\]
Applying Lemma \ref{L:Eq.estrella} we get $\bar g_3\equiv0$ and
hence $g_3\equiv0$. Consequently $\pd{h_3}{x_6}=0$ and the lemma
follows in this case.

\smallskip

If $m>2$ then from \eqref{VIII1} we have
$x_3\left|\pd{h_3}{x_6}\right.$ and so $h_3=h_3(x_4-x_5)$. Now from
equation \eqref{VIII1} on $x_3=0$ we get an equation of type
\eqref{Eq:estrella}, hence applying Lemma \ref{L:Eq.estrella} we get
$g_3\equiv0$ and the lemma follows in this case.

\smallskip

If $m=2$ then from \eqref{VIII1} we have
\[
\begin{split}
2(x_4 &+ x_5 - x_6)g_3+x_3(x_4 + x_5 - x_6)\pd{g_3}{x_3}+x_3^2
\pd{g_3}{x_6}\\
&+\frac{k-1}{4}F_{12}\left(\pd{g_3}{x_4}+\pd{g_3}{x_5}+\pd{g_3}
{x_6}\right)+\pd{h_3}{x_6}=0.
\end{split}
\]
Let $\bar g_3=g_3|_{x_3=0}\not\equiv0$. On $x_3=0$ we obtain
\[
2(x_4 + x_5 - x_6)\bar g_3+\frac{k-1}{4}F_{123}\left(\pd{\bar
g_3}{x_4}+\pd{\bar g_3}{x_5}+\pd{\bar
g_3}{x_6}\right)+\pd{h_3}{x_6}=0.
\]
Applying Lemma \ref{L:Eq.dificil} swapping $x_5$ and $x_6$ we get $\pd{h_3}{x_6}=0$ and $\bar
g_3\equiv0$. Hence $h_3=h_3(x_4-x_5)$, $g_3\equiv0$ and the lemma
follows in this case.
\end{proof}

The second lemma shows that $h_1\equiv0$.

\begin{lemma}\label{L:g2+h2}
For system \eqref{VIII-IXx1} we have that $g_2\equiv 0$ and
$h_2\equiv0$.
\end{lemma}

\begin{proof}
Suppose that $g_2\not\equiv0$. We recall that
$h_1=h_2(x_4-x_5)+x_2^lg_2$, with $l\in\N$ and $x_2\nmid g_2$. As
$h_1$ is a first integral of system \eqref{VIII-IXx1}, we have
\[
\begin{split}
x_2^l&\left[l(x_4-x_5+x_6)g_2+x_2(x_4-x_5+x_6)\pd{g_2}{x_2}+
x_3(x_4+x_5-x_6)\pd{g_2}{x_3}+x_2(x_2-n_3x_3)\pd{g_2}{x_5}\right.\\
&-n_3x_3(x_2-n_3x_3)\pd{g_2}{x_6}\left.+\frac{k-1}{4}F_1\left(\pd{g_2}{x_4}+\pd{g_2}{x_5}+
\pd{g_2}{x_6}\right)\right]+x_2(x_2-n_3x_3)\pd{h_2}{x_5}=0.
\end{split}
\]
We distinguish two cases depending on the value of $l$. If $l>1$
then $x_2\left|\pd{h_2}{x_5}\right.$ and hence
$\pd{h_2}{x_5}\equiv0$, which means that $h_2\equiv0$. Substituting
in the equation above we have
\[
\begin{split}
l(x_4&-x_5+x_6)g_2+x_2(x_4-x_5+x_6)\pd{g_2}{x_2}+x_3(x_4+x_5-x_6)
\pd{g_2}{x_3}\\
&+x_2(x_2-n_3x_3)\pd{g_2}{x_5}-n_3x_3(x_2-n_3x_3)\pd{g_2}{x_6}+\frac{k-1}{4}F_1\left(\pd{g_2}{x_4}+\pd{g_2}{x_5}+\pd{g_2}{x_6}
\right)=0.
\end{split}
\]
The usual arguments lead to equation \eqref{Eq:estrella}, hence we
obtain $g_2\equiv0$ by Lemma \ref{L:Eq.estrella}.

\smallskip

If $l=1$ then we have
\[
\begin{split}
(x_4&-x_5+x_6)g_2+x_2(x_4-x_5+x_6)\pd{g_2}{x_2}+x_3(x_4+x_5-x_6)
\pd{g_2}{x_3}+x_2(x_2-n_3x_3)\pd{g_2}{x_5}\\
&-n_3x_3(x_2-n_3x_3)\pd{g_2}{x_6}+\frac{k-1}{4}F_1\left(\pd{g_2}{x_4}+\pd{g_2}{x_5}+\pd{g_2}{x_6}
\right)+(x_2-n_3x_3)\pd{h_2}{x_5}=0.
\end{split}
\]
Let $\bar g_2=g_2|_{x_2=0}\not\equiv0$. On $x_2=0$ we have
\begin{equation}\label{Eq:x2=0}
\begin{split}
(x_4&-x_5+x_6)\bar g_2+x_3(x_4+x_5-x_6)\pd{\bar g_2}{x_3}+
x_3^2\pd{\bar g_2}{x_6}\\
&+\frac{k-1}{4}F_{12}\left(\pd{\bar g_2}{x_4}+\pd{\bar g_2}{x_5}+
\pd{\bar g_2}{x_6}\right)-n_3x_3\pd{h_2}{x_5}=0.
\end{split}
\end{equation}
Write $\bar g_2=x_3^mg_3\not\equiv0$, with $m\in\N\cup\{0\}$ and
$x_3\nmid g_3$. Then
\[
\begin{split}
x_3^m&\left[((x_4-x_5+x_6)+m(x_4+x_5-x_6))g_3+x_3(x_4+x_5-x_6)
\pd{g_3}{x_3}\right.\\
&\left.+x_3^2\pd{g_3}{x_6}+\frac{k-1}{4}F_{12}\left(\pd{g_3}{x_4}+
\pd{g_3}{x_5}+\pd{g_3}{x_6}\right)\right]-n_3x_3\pd{h_2}{x_5}=0.
\end{split}
\]
Now we distinguish three cases depending on the value of $m$. If
$m=0$ then we are in \eqref{Eq:x2=0} again and the usual arguments
lead to $g_2\equiv0$ and $h_2\equiv0$.

\smallskip

If $m>1$ then $x_3\left|\pd{h_2}{x_5}\right.$ and hence
$\pd{h_2}{x_5}\equiv0$, which means that $h_2\equiv0$. Then we have
\[
\begin{split}
((x_4&-x_5+x_6)+m(x_4+x_5-x_6))g_3+x_3(x_4+x_5-x_6)\pd{g_3}{x_3}\\
&+x_3^2\pd{g_3}{x_6}+\frac{k-1}{4}F_{12}\left(\pd{g_3}{x_4}+
\pd{g_3}{x_5}+\pd{g_3}{x_6}\right)=0.
\end{split}
\]
The usual arguments finish the proof in this case.

\smallskip

Finally if $m=1$ then we have
\[
2x_4g_3+x_3(x_4+x_5-x_6)\pd{g_3}{x_3}+x_3^2\pd{g_3}{x_6}+\frac{k-1}{4}F_{12}\left(\pd{g_3}{x_4}+\pd{g_3}{x_5}+\pd{g_3}{x_6}\right)-n_3\pd{h_2}{x_5}=0.
\]
Let $\bar g_3=g_3|_{x_3=0}\not\equiv0$. On $x_3=0$ we have
\[
2x_4\bar g_3+\frac{k-1}{4}F_{123}\left(\pd{\bar g_3}{x_4}+\pd{\bar
g_3}{x_5}+\pd{\bar g_3}{x_6}\right)-n_3\pd{h_2}{x_5}=0.
\]
As $h_2=h_2(x_4-x_5)$ is a homogeneous polynomial of degree $n$, we
have $h_2=a_0(x_4-x_5)^n$. Hence
\[
2x_4\bar g_3+\frac{k-1}{4}F_{123}\left(\pd{\bar g_3}{x_4}+\pd{\bar
g_3}{x_5}+\pd{\bar g_3}{x_6}\right)+n_3a_0n(x_4-x_5)^{n-1}=0.
\]
On $x_4=0$ we have
\[
\frac{k-1}{4}(x_5-x_6)^2\left.\left(\pd{\bar g_3}{x_4}+\pd{\bar
g_3}{x_5}+\pd{\bar
g_3}{x_6}\right)\right|_{x_4=0}+n_3a_0n(-x_5)^{n-1}=0,
\]
which means that $a_0=0$. Therefore $h_2\equiv0$. The equation is
now of type \eqref{Eq:estrella} and leads to $g_2\equiv0$ by Lemma
\ref{L:Eq.estrella}.

\smallskip

All the subcases are considered and the proof of the lemma is
finished.
\end{proof}

The last lemma shows that $g_1\equiv0$ and therefore that
$h\equiv0$.

\begin{lemma}\label{L:g1=0}
For system \eqref{VIII-IX} we have that $g_1\equiv 0$.
\end{lemma}

\begin{proof}
Suppose that $g_1\not\equiv0$. We recall that $h=x_1^jg_1$, with
$j\in\N$ and $x_1\nmid g_1$, is a first integral of system
\eqref{VIII-IX}. Then
\[
\begin{split}
j(-x_4&+x_5+x_6)g_1+x_1(-x_4+x_5+x_6)\pd{g_1}{x_1}+x_2 (x_4 - x_5 +
x_6)\pd{g_1}{x_2}\\
&+x_3 (x_4 + x_5 -
x_6)\pd{g_1}{x_3}+x_1(x_1-x_2-n_3x_3)\pd{g_1}{x_4}\\
&+x_2(-x_1+x_2-n_3x_3)\pd{g_1}{x_5}+n_3x_3(-x_1-x_2+n_3x_3)\pd{g_1}{x_6}\\
&+\frac{k-1}{4}F\left(\pd{g_1}{x_4}+\pd{g_1}{x_5}+\pd{g_1}{x_6}\right)=0.
\end{split}
\]
Let $\bar g_1=g_1|_{x_1=0}\not\equiv0$. On $x_1=0$ we have
\[
\begin{split}
j(-x_4&+x_5+x_6)\bar g_1+x_2 (x_4 - x_5 + x_6)\pd{\bar g_1}{x_2}+x_3
(x_4 + x_5 - x_6)\pd{\bar g_1}{x_3}\\
&+x_2(x_2-n_3x_3)\pd{\bar g_1}{x_5}+n_3x_3(-x_2+n_3x_3)\pd{\bar
g_1}{x_6}+\frac{k-1}{4}F_1\left(\pd{\bar g_1}{x_4}+\pd{\bar
g_1}{x_5}+\pd{\bar g_1}{x_6}\right)=0.
\end{split}
\]
Write $\bar g_1=x_2^lg_2\not\equiv0$, with $l\in\N\cup\{0\}$ and
$x_2\nmid g_2$. We get
\[
\begin{split}
(j(-x_4&+x_5+x_6)+l(x_4 - x_5 + x_6))g_2+x_2 (x_4 - x_5 +
x_6)\pd{g_2}{x_2}\\
&+x_3 (x_4 + x_5 - x_6)\pd{g_2}{x_3}+x_2(x_2-n_3x_3)\pd{g_2}{x_5}\\
&+n_3x_3(-x_2+n_3x_3)\pd{g_2}{x_6}+\frac{k-1}{4}F_1\left(\pd{g_2}{x_4}
+\pd{g_2}{x_5}+\pd{g_2}{x_6}\right)=0.
\end{split}
\]
Let $\bar g_2=g_2|_{x_2=0}\not\equiv0$. On $x_2=0$ we have
\[
\begin{split}
(j(-x_4&+x_5+x_6)+l(x_4 - x_5 + x_6))\bar g_2+x_3 (x_4 + x_5 -
x_6)\pd{\bar g_2}{x_3}\\ &+x_3^2\pd{\bar
g_2}{x_6}+\frac{k-1}{4}F_{12}\left(\pd{\bar g_2}{x_4}+\pd{\bar
g_2}{x_5}+\pd{\bar g_2}{x_6}\right)=0.
\end{split}
\]
Write $\bar g_2=x_3^mg_3\not\equiv0$, with $m\in\N\cup\{0\}$ and
$x_3\nmid g_3$. We get
\[
\begin{split}
(j(-x_4&+x_5+x_6)+l(x_4 - x_5 + x_6)+m(x_4 + x_5 - x_6))g_3\\
&+x_3 (x_4 + x_5 - x_6)\pd{g_3}{x_3}+x_3^2\pd{g_3}{x_6}+\frac{k-1}{4}F_{12}\left(\pd{g_3}{x_4}+\pd{g_3}{x_5}+
\pd{g_3}{x_6}\right)=0.
\end{split}
\]
Let $\bar g_3=g_3|_{x_3=0}\not\equiv0$. On $x_3=0$ we have
\[
\begin{split}
(j(-x_4&+x_5+x_6)+l(x_4 - x_5 + x_6)+m(x_4 + x_5 - x_6))\bar g_3\\
&+\frac{k-1}{4}F_{123}\left(\pd{\bar g_3}{x_4}+\pd{\bar g_3}{x_5}+
\pd{\bar g_3}{x_6}\right)=0.
\end{split}
\]
We can apply Lemma \ref{L:Eq.estrella}. Hence the lemma follows.
\end{proof}

After Lemma \ref{L:g1=0}, we get $h\equiv0$. Thus the proof of
statement (d) of Theorem \ref{main} is finished.

\end{document}